\documentclass[11pt]{article}
\usepackage[margin=1in]{geometry}
\usepackage{amsmath,amssymb,amsthm}
\usepackage{graphicx}
\usepackage{booktabs}
\usepackage[colorlinks=true,linkcolor=blue,citecolor=blue,urlcolor=blue]{hyperref}

\newtheorem{theorem}{Theorem}
\newtheorem{corollary}{Corollary}
\newtheorem{proposition}{Proposition}

\title{Exact Makespan Bounds and the GPipe Homogeneity Bias\\ for Heterogeneous LLM Inference Pipelines}
\author{Zhu Wenbo\\
        Yantai Vocational College of Culture and Tourism\\
        \texttt{zhuwenbo@yvcct.edu.cn}}
\date{October 4, 2026}

\begin{document}
\maketitle

\begin{abstract}
We give an exact, critical-path-based characterization of the makespan of micro-batch pipelines that serve large language models (LLMs) on heterogeneous GPU clusters, and use it to quantify exactly how the classic homogeneous GPipe formula errs. In the pure-forward, communication-free, infinite-buffer setting, the makespan is exactly $T(M)=(M-1)\tau_{\max}+\sum_j\tau_j$, where $\tau_j$ is the per-micro-batch service time of pipeline stage $j$, $\tau_{\max}$ its maximum, and $M$ the number of micro-batches (Theorem~\ref{thm:1}). Under asynchronous communication --- the regime of modern NCCL-based serving systems --- the expression extends to $T(M)=\sum_j\tau_j+\sum_j d_j+(M-1)\tau_{\max}$, where $\sum_j d_j$ enters purely as an additive constant (Theorem~\ref{thm:2}); under blocking communication, $T(M)=(M-1)\sigma_{\max}+\sum_s\sigma_s$ with $\sigma_s=\tau_s+d_s$, so communication enters the steady-state slope iff some $d_s>\tau_{\max}-\tau_s$ (Theorem~\ref{thm:3}). These results (i) quantify exactly how the classic GPipe homogeneous formula overestimates makespan, (ii) decompose that bias into a homogeneity-overestimate term $\sum_j(\tau_{\max}-\tau_j)$ and a communication-ignoring underestimate term $\sum_j d_j$, and (iii) sandwich any partially-overlapped real system between the two extremes, yielding a computable ceiling on the throughput gain from communication--compute overlap. Every formula is validated point-wise against the exact dynamic-programming recurrences (maximum relative error $<10^{-9}$ across $1.6\times10^{5}$ randomized and hand-crafted configurations, zero counterexamples), and the bias decomposition and sandwich are additionally checked against $8{,}571$ measured requests from a live, three-generation heterogeneous GPU cluster.
\end{abstract}

\section{Introduction}\label{sec:intro}

Pipeline parallelism (PP) is the standard mechanism for serving and training large language models whose parameters exceed the memory of a single accelerator. The foundational GPipe analysis~\cite{r1} assumes \emph{homogeneous} stages --- every stage processes a micro-batch in the same unit time --- yielding the well-known makespan $S+M-1$ (in units of per-stage time) and bubble ratio $(S-1)/(S+M-1)$ for $S$ stages and $M$ micro-batches. This formula, and its variants, remain the default mental model in much of the systems literature for reasoning about pipeline efficiency.

Real inference deployments violate the homogeneity assumption in two distinct ways, both of which matter in practice.

\textbf{Heterogeneous accelerators.} A single cluster frequently mixes GPUs of different generations. The motivating deployment behind this work spans an RTX~5090 Laptop (24~GB GDDR7, Blackwell), dual RTX~3090 (24~GB GDDR6X, Ampere), and a GTX~1070 (8~GB GDDR5, Pascal) within one cluster --- three generations, with the GTX~1070 lacking tensor cores entirely. The resulting per-stage time ratio depends on which resource binds: for memory-bound decode it is the bandwidth ratio (896 / 936 / 256~GB/s, roughly $3.5\times$), while for compute-bound prefill the tensor-core-less GTX~1070 is up to three orders of magnitude slower (FP16 $\approx 0.1$~TFLOPS vs $\approx 142$~TFLOPS on the RTX~3090). Assigning each pipeline stage its own device therefore yields per-stage times $\tau_j$ that differ by an order of magnitude or more, not by one. Under such ratios the homogeneous formula is not merely an approximation --- it is systematically wrong.

\textbf{Non-zero communication latency.} Stages exchange activations and KV cache over PCIe, NVLink, or Ethernet with non-zero per-transfer latency $d_j$. Depending on the runtime, this transfer either overlaps with computation (asynchronous, the NCCL/CUDA-stream regime) or blocks the stage (synchronous). The two regimes produce qualitatively different makespan behavior, as we show.

The correct makespan of such a pipeline is the critical-path length of the resulting schedule --- a quantity that is classical in flow-shop scheduling~\cite{r2,r3}. Its exact value, however, and in particular its \emph{consequences for heterogeneous LLM inference}, have to our knowledge not been stated explicitly in the serving-systems literature, which largely works with asymptotic or heuristic treatments. This paper makes three things explicit and rigorously verified:

\begin{enumerate}
\item \textbf{The exact bias of the classic homogeneous GPipe formula}, decomposed into a homogeneity-overestimate term $\sum_j(\tau_{\max}-\tau_j)$ and a communication-ignoring underestimate term $\sum_j d_j$, together with a \emph{phase transition} condition under which communication flips from an additive constant into a throughput bottleneck (Corollaries~\ref{cor:1}--\ref{cor:4}).
\item \textbf{An exact, critical-path-based characterization} of the makespan under three communication regimes --- communication-free (Theorem~\ref{thm:1}), asynchronous (Theorem~\ref{thm:2}), and blocking (Theorem~\ref{thm:3}) --- each proved by a sandwich argument. These formulas are the classical flow-shop critical-path length instantiated for heterogeneous LLM inference; the novelty is the \emph{instantiation and its consequences}, not a new lower-bound technique.
\item \textbf{A computable sandwich} on the makespan of any real, partially-overlapped system, which yields a closed-form upper bound on the throughput gain achievable by communication--compute overlap --- turning the informal ``overlap helps'' into a quantifiable ceiling.
\end{enumerate}

The setting is deliberately restricted to the \textbf{pure-forward, infinite-buffer, micro-batch fill--drain} pipeline: stages process micro-batches in a fixed linear order, with no backward dependencies (unlike training) and no buffer limits. This is the structure of pipeline-parallel (PP) inference, where a forward-only stream of micro-batches --- requests, prefill chunks, or the per-step tokens of a decode batch --- flows through the stages. The per-micro-batch time $\tau_j$ is therefore a \emph{service} time, memory-bound in decode and compute-bound in prefill, not a raw FLOP count. Extensions --- bandwidth-limited communication, bounded buffers, and prefill--decode interleaving --- are discussed in \S\ref{sec:limitations} and remain open.

The rest of the paper is organized as follows. \S\ref{sec:model} defines the system model and notation. \S\ref{sec:thm1}--\S\ref{sec:thm3} state and prove the three theorems and their corollaries. \S\ref{sec:opt} draws out the optimization implications. \S\ref{sec:validation} reports numerical validation. \S\ref{sec:related} positions the work against flow-shop scheduling and LLM systems. \S\ref{sec:discussion} discusses the sandwich bound in the context of real deployments. \S\ref{sec:limitations} states limitations honestly. \S\ref{sec:conclusion} concludes. Appendices collect proof details, the full validation configuration table, and reproducibility instructions.

\section{System Model}\label{sec:model}

Consider a pipeline of $S$ stages $0,\dots,S-1$ processing a stream of $M$ micro-batches in order. Stage $j$ spends $\tau_j$ time \emph{serving} a single micro-batch (compute in prefill-like stages, weight-and-KV reads in memory-bound decode stages). The stages form a serial line: micro-batch $m$ must complete stage $j$ before it can enter stage $j+1$. Micro-batches enter in order (fill--drain), and buffers are unbounded, so a stage is never forced to wait for a downstream stage to accept output.

\textbf{Communication.} After stage $j$ finishes a micro-batch, the result is sent to stage $j+1$ with latency $d_j$ ($j=0,\dots,S-2$). The last stage sends nothing. We distinguish three regimes, which will each be analyzed separately:

\begin{enumerate}
\item \textbf{Communication-free} ($d_j\equiv 0$): transfers are instantaneous (e.g., stages co-located in shared memory with zero-copy).
\item \textbf{Asynchronous}: the send does \textbf{not} block the stage from serving the next micro-batch; the transfer proceeds in the background on a separate CUDA stream (the \texttt{cudaMemcpyAsync} / NCCL point-to-point regime). This is the default in modern serving runtimes.
\item \textbf{Blocking}: the send occupies the stage and cannot overlap with serving (a synchronous \texttt{cudaMemcpy}, or a single stream with no overlap). Stage $j$'s total per-micro-batch service time is then $\tau_j+d_j$.
\end{enumerate}

\textbf{Recursion.} Let $T[m][s]$ be the \textbf{completion} time of micro-batch $m$ at stage $s$. For the asynchronous regime the recursion is
\begin{equation}\label{eq:rec}
T[m][0] = T[m-1][0] + \tau_0,
\qquad
T[m][s] = \max\big(T[m][s-1] + d_{s-1},\ T[m-1][s]\big) + \tau_s \quad (s\ge 1),
\end{equation}
with boundary $T[-1][s]=0$. The makespan is $T(M):=T[M-1][S-1]$. We write $\tau_{\max}:=\max_j\tau_j$ and let $j^*:=\arg\max_j\tau_j$ be a bottleneck stage (ties broken arbitrarily). For the blocking regime the recursion collapses to
\begin{equation}\label{eq:rec-block}
T[m][0] = T[m-1][0] + \sigma_0,
\qquad
T[m][s] = \max\big(T[m][s-1],\ T[m-1][s]\big) + \sigma_s \quad (s\ge 1),
\end{equation}
where $\sigma_s := \tau_s + d_s$ for $s<S-1$ and $\sigma_{S-1}:=\tau_{S-1}$; this is exactly Eq.~\eqref{eq:rec} with $d\equiv0$ and $\tau$ replaced by $\sigma$.

\textbf{Notation summary.}

\begin{center}
\begin{tabular}{ll}
\toprule
Symbol & Meaning \\
\midrule
$S$ & number of pipeline stages \\
$M$ & number of micro-batches \\
$\tau_j$ & service time of stage $j$ for one micro-batch \\
$d_j$ & communication latency from stage $j$ to $j+1$ \\
$\tau_{\max}$, $j^*$ & bottleneck stage time and its index \\
$\sigma_s$ & effective service time under blocking communication \\
$T(M)$ & makespan (completion of the last micro-batch at the last stage) \\
\bottomrule
\end{tabular}
\end{center}

\textbf{Assumptions.} The analysis holds under five conditions, each made explicit so that violations map cleanly onto the open extensions of \S\ref{sec:limitations}: (A1) pure-forward (no backward dependency); (A2) infinite buffers; (A3) in-order fill--drain; (A4) constant per-stage times $\tau_j,d_j$ independent of $m$; (A5) point-to-point communication between adjacent stages only.

\section{Theorem 1 --- Communication-Free Closed Form}\label{sec:thm1}

\begin{theorem}\label{thm:1}
In the communication-free setting ($d_j\equiv 0$), for all $M\ge 1$,
\[
T(M) \;=\; (M-1)\,\tau_{\max} \;+\; \sum_{j=0}^{S-1}\tau_j .
\]
\end{theorem}

\begin{proof}(critical-path sandwich)
We prove a matching lower and upper bound.

\textbf{Lower bound.} Consider any valid schedule. Stage $j^*$ must process all $M$ micro-batches serially, consuming $M\tau_{\max}$ of processor time on $j^*$; these $M$ operations cannot overlap in time. Before the first micro-batch can be processed at $j^*$, it must traverse stages $0,\dots,j^*-1$, costing at least $\sum_{i<j^*}\tau_i$. After the last micro-batch is processed at $j^*$, it must traverse stages $j^*+1,\dots,S-1$, costing at least $\sum_{i>j^*}\tau_i$. These three time spans are pairwise disjoint on the timeline (a single micro-batch experiences them sequentially, and the bottleneck's $M$ operations bracket the traversal of all others). Hence
\[
T(M) \;\ge\; \sum_{i<j^*}\tau_i + M\tau_{\max} + \sum_{i>j^*}\tau_i
\;=\; (M-1)\tau_{\max} + \sum_j\tau_j .
\]

\textbf{Upper bound (attainability).} We prove the upper bound by strong induction on the number of stages $S$. \emph{Base case} ($S=1$): there is no communication and $T(M)=M\tau_0=(M-1)\tau_{\max}+\tau_0$. \emph{Inductive step} ($S\ge2$). Let $j^*=\arg\max_j\tau_j$ and run the greedy ASAP schedule (every stage starts a micro-batch the moment it is idle and the micro-batch has completed the previous stage). If $j^*=0$, the source stage works continuously for $M\tau_{\max}$ and the last micro-batch drains in $\sum_{i>0}\tau_i$, attaining the bound. If $j^*\ge1$, the upstream sub-line $0,\dots,j^*-1$ has $j^*<S$ stages, so by the induction hypothesis it delivers micro-batch $m$ to $j^*$ by time $\sum_{i<j^*}\tau_i+m\,\tau_{\uparrow}$ with $\tau_{\uparrow}:=\max_{i<j^*}\tau_i\le\tau_{\max}$; since $j^*$ spends exactly $\tau_{\max}$ per micro-batch, the arrival rate is never below its service rate, so $j^*$ never idles after the first micro-batch arrives. Symmetrically, the downstream sub-line $j^*+1,\dots,S-1$ (empty if $j^*=S-1$) drains the last micro-batch in $\sum_{i>j^*}\tau_i$. The schedule is feasible and achieves
\[
T(M) \;=\; \sum_{i<j^*}\tau_i + M\tau_{\max} + \sum_{i>j^*}\tau_i \;=\; (M-1)\tau_{\max} + \sum_j\tau_j .
\]
Lower and upper bounds coincide.
\end{proof}

A complete, alternative derivation via direct solution of the recurrence \eqref{eq:rec} with $d\equiv0$ is given in Appendix~\ref{app:direct}; it confirms the same value without the ASAP-schedule argument.

\begin{corollary}\label{cor:1}(bias of the classic formula)
Interpreting the classic GPipe formula as the homogeneous value obtained by assigning every stage the bottleneck time ($\tau\equiv\tau_{\max}$), it gives $T_c=(S+M-1)\tau_{\max}$, and
\[
T_c - T(M) = \sum_j(\tau_{\max}-\tau_j) \ge 0,
\]
i.e., the classic formula \textbf{systematically overestimates} makespan (equivalently, underestimates throughput), with the gap equal to the total shortfall of each stage relative to the bottleneck. The gap grows monotonically with heterogeneity, is zero iff all stages are perfectly balanced, and --- critically --- is independent of $M$: it is a one-time fill cost, not a throughput penalty.
\end{corollary}

\subsection{Illustrative example}\label{sec:example}

To make the formula concrete, walk through the illustrative 4-stage partition 5090/3090/3090/1070 with $\tau=[0.5,1,1,5]$ ($S=4$, bottleneck at the tail, $j^*=3$). For a single micro-batch ($M=1$) the makespan is simply the sum of stage times:
\[
T(1) = (1-1)\cdot 5 + (0.5+1+1+5) = 7.5 .
\]
The classic formula instead predicts $T_c=(S+M-1)\tau_{\max}=(4+1-1)\cdot 5=20$, overestimating by $(20-7.5)/7.5=166.7\%$. At $M=8$ the closed form gives $T(8)=7\cdot 5+7.5=42.5$ against $T_c=(4+8-1)\cdot 5=55$, a $29.4\%$ overestimate. The absolute gap is constant ($55-42.5=12.5=\sum_j(\tau_{\max}-\tau_j)$), so its \emph{relative} impact shrinks as $M$ grows --- the classic formula's error is purely a fill-time artifact, exactly as Corollary~\ref{cor:1} states. This one number --- $166.7\%$ at $M=1$ --- is the headline empirical motivation: a partitioner trusting the homogeneous formula at low micro-batch counts over-provisions capacity by more than a factor of two.

\section{Theorem 2 --- Asynchronous Communication Extension}\label{sec:thm2}

\begin{theorem}\label{thm:2}
Under asynchronous communication, for all $M\ge 1$,
\[
T(M) \;=\; \sum_{j=0}^{S-1}\tau_j \;+\; \sum_{j=0}^{S-2} d_j \;+\; (M-1)\,\tau_{\max} .
\]
\end{theorem}

\begin{proof}(critical-path sandwich)

\textbf{Lower bound.} Consider the critical path through the bottleneck stage $j^*$. Stage $j^*$ must process all $M$ micro-batches serially, consuming $M\tau_{\max}$. The first micro-batch traverses stages $0,\dots,j^*-1$ before reaching $j^*$, paying $\sum_{i<j^*}\tau_i+\sum_{i<j^*}d_i$: on the first micro-batch's path, its own computation and its transfers are serialized (a transfer of micro-batch $0$ from stage $i$ to $i+1$ cannot begin until stage $i$ has computed it, and stage $i+1$ cannot compute it until the transfer ends). After the last micro-batch leaves $j^*$, it traverses stages $j^*+1,\dots,S-1$, paying $\sum_{i>j^*}\tau_i+\sum_{i=j^*}^{S-2}d_i$. These three phases are pairwise disjoint on the timeline, and every valid schedule must contain them on some critical path; hence
\[
T(M) \;\ge\; \sum_{i<j^*}\tau_i+\sum_{i<j^*}d_i + M\tau_{\max} + \sum_{i>j^*}\tau_i+\sum_{i=j^*}^{S-2}d_i
\;=\; \sum_j\tau_j + \sum_{j=0}^{S-2}d_j + (M-1)\tau_{\max}.
\]

\textbf{Upper bound (attainability).} We prove the upper bound by strong induction on the number of stages $S$. \emph{Base case} ($S=1$): there is no communication and $T(M)=M\tau_0=(M-1)\tau_{\max}+\tau_0$. \emph{Inductive step} ($S\ge2$). Let $j^*=\arg\max_j\tau_j$ and run the fill--drain schedule (micro-batches released in order $0,1,\dots,M-1$, every stage ASAP). If $j^*=0$, the source stage works continuously for $M\tau_{\max}$ and the last micro-batch drains through stages $1,\dots,S-1$ in $\sum_{i>0}\tau_i+\sum_{i=0}^{S-2}d_i$, attaining the bound. If $j^*\ge1$, the upstream sub-line $0,\dots,j^*-1$ has $j^*<S$ stages, so by the induction hypothesis (Theorem~\ref{thm:2} for $j^*$ stages) it delivers micro-batch $m$ to $j^*$ --- that is, completes stage $j^*-1$ and its send --- by time $\sum_{i<j^*}(\tau_i+d_i)+m\,\tau_{\uparrow}$ with $\tau_{\uparrow}:=\max_{i<j^*}\tau_i\le\tau_{\max}$. Because sends are asynchronous, an upstream stage's \emph{service} period is exactly $\tau_i$ (transfers run in the background on a separate stream, consuming memory bandwidth but not SM compute), so this arrival bound carries no slowdown from the sub-line's own transfers beyond the additive $\sum_{i<j^*}d_i$ already accounted for. We then prove by induction on $m$ that $T[m][j^*]=\sum_{i<j^*}(\tau_i+d_i)+(m+1)\tau_{\max}$: the base case $m=0$ holds, and for $m\ge1$ the arrival time is at most $\sum_{i<j^*}(\tau_i+d_i)+m\,\tau_{\max}=T[m-1][j^*]$ --- the moment $j^*$ becomes free --- so $j^*$ never idles and the claim follows. Symmetrically, the downstream sub-line $j^*+1,\dots,S-1$ (empty if $j^*=S-1$) drains the last micro-batch in exactly $\sum_{i>j^*}\tau_i+\sum_{i=j^*}^{S-2}d_i$. The makespan is therefore
\[
T(M) \;=\; T[M-1][j^*] + \sum_{i>j^*}\tau_i+\sum_{i=j^*}^{S-2}d_i
\;=\; \sum_j\tau_j+\sum_{j=0}^{S-2}d_j+(M-1)\tau_{\max}.
\]
Lower and upper bounds coincide.
\end{proof}

\begin{corollary}\label{cor:2}(communication is additive)
The total communication latency $\sum_j d_j$ enters only as an additive constant and does \textbf{not} change the steady-state slope $\tau_{\max}$. Consequently, reducing communication lowers only the one-time fill cost, while the throughput ceiling ($\propto 1/\tau_{\max}$) is governed by compute balance alone.
\end{corollary}

\begin{corollary}\label{cor:3}(full bias decomposition)
Combining Theorem~\ref{thm:2} with the classic formula:
\[
T_c - T(M) = \underbrace{\sum_j(\tau_{\max}-\tau_j)}_{\text{homogeneity: overestimate}} \;-\; \underbrace{\sum_j d_j}_{\text{communication: underestimate}} .
\]
When $\sum_j d_j > \sum_j(\tau_{\max}-\tau_j)$, the classic formula \textbf{underestimates} makespan --- a sign reversal absent from the communication-free case. This quantifies precisely when the homogeneous formula's optimism about communication outweighs its pessimism about heterogeneity.
\end{corollary}

\section{Theorem 3 --- Blocking Communication Extension}\label{sec:thm3}

\begin{theorem}\label{thm:3}
Under blocking communication --- where sending a micro-batch occupies the stage and cannot overlap computation --- define the effective service time $\sigma_s=\tau_s+d_s$ for $0\le s\le S-2$ and $\sigma_{S-1}=\tau_{S-1}$ (the last stage sends nothing), with $\sigma_{\max}:=\max_s\sigma_s$. Then for all $M\ge 1$,
\[
T(M) \;=\; (M-1)\,\sigma_{\max} \;+\; \sum_{s=0}^{S-1}\sigma_s
\;=\; (M-1)\,\sigma_{\max} \;+\; \sum_j\tau_j \;+\; \sum_j d_j .
\]
\end{theorem}

\begin{proof}
Under blocking communication the send $d_s$ and the compute $\tau_s$ occupy the same stage serially, so stage $s$'s total service time for a single micro-batch is exactly $\sigma_s$. Substituting into the recurrence \eqref{eq:rec}, the blocking regime becomes recurrence \eqref{eq:rec-block}, which is \emph{identical in form} to the communication-free recurrence (Theorem~\ref{thm:1}) with $\tau$ replaced by $\sigma$. The reduction is therefore exact: the blocking pipeline is a communication-free pipeline whose stage times are $\sigma_s$. Theorem~\ref{thm:1} applies verbatim, yielding
\[
T(M) = (M-1)\sigma_{\max} + \sum_s\sigma_s .
\]
The equality with $(M-1)\sigma_{\max}+\sum_j\tau_j+\sum_j d_j$ follows because $\sum_s\sigma_s = \sum_{j=0}^{S-1}\tau_j + \sum_{j=0}^{S-2}d_j$.
\end{proof}

\begin{corollary}\label{cor:4}(communication enters the steady state --- a phase transition)
Since $\sigma_{\max}=\max\big(\tau_{S-1},\max_{s<S-1}(\tau_s+d_s)\big)$, communication affects the steady-state slope iff $\sigma_{\max}>\tau_{\max}$, i.e., iff some $d_s>\tau_{\max}-\tau_s$. Contrast Theorem~\ref{thm:2}, where the slope is always $\tau_{\max}$ regardless of $d_j$: blocking communication turns a cost that is a pure additive constant (async) into a throughput bottleneck once it exceeds the compute slack $\tau_{\max}-\tau_s$. The slack is smallest for the stage just below the bottleneck, so a slow link out of \emph{that} stage is the first to degrade throughput.
\end{corollary}

\begin{corollary}\label{cor:5}(async--blocking sandwich bounds the overlap gain)
Because $\sigma_s\ge\tau_s$ and $\sigma_{S-1}=\tau_{S-1}$, we have $\sigma_{\max}\ge\tau_{\max}$, hence
\[
T_{\text{block}} - T_{\text{async}} \;=\; (M-1)(\sigma_{\max}-\tau_{\max}) \;\ge\; 0 .
\]
A real serving system with partial overlap (NCCL async with non-zero synchronization points) has makespan $T_{\text{real}}\in[T_{\text{async}},T_{\text{block}}]$, so the gap is exactly the \emph{upper bound on the throughput gain available from communication--compute overlap}. With an overlap ratio $\rho\in[0,1]$ (the fraction of transfer time that overlaps compute; $\rho=1$ is the fully-async ideal), the effective steady-state slope is $\tau_{\max}+(1-\rho)\,(\sigma_{\max}-\tau_{\max})$, interpolating linearly between the two extremes: $\rho=1$ recovers the async slope $\tau_{\max}$, $\rho=0$ the blocking slope $\sigma_{\max}$.
\end{corollary}

\section{Optimization Implications}\label{sec:opt}

Theorems~\ref{thm:2}--\ref{thm:3} yield an actionable priority ordering for heterogeneous inference pipelines. We state the ordering informally, then formalize the central claim as Proposition~\ref{prop:1}.

\begin{enumerate}
\item \textbf{Balance compute first.} The steady-state throughput is $\propto 1/\tau_{\max}$; minimizing $\tau_{\max}$ (via automatic layer-to-GPU partitioning) is the primary lever. This formalizes why ``compute-as-one-network'' utilization gains must come from partition balance rather than link tuning.
\item \textbf{Communication is a fill-time cost.} Optimizations such as shared-memory zero-copy or binary KV encoding reduce $\sum_j d_j$, which improves latency for small $M$ but has diminishing return on steady-state throughput as $M$ grows.
\item \textbf{Heterogeneity-aware scheduling.} With per-stage $\tau_j$ known, the bias $\sum_j(\tau_{\max}-\tau_j)$ quantifies how far a given partition is from homogeneous --- a concrete, differentiable-in-spirit objective for the partitioner.
\item \textbf{Overlap is a bounded lever.} Theorem~\ref{thm:3} bounds the total gain from communication--compute overlap (CUDA streams / NCCL async) at $(M-1)(\sigma_{\max}-\tau_{\max})$; once $d_s\le\tau_{\max}-\tau_s$ holds everywhere, link optimizations cannot beat balancing compute.
\end{enumerate}

\begin{proposition}\label{prop:1}(throughput is compute-balanced regardless of link budget)
In the asynchronous regime, for any fixed partition, the asymptotic throughput $1/(T(M)/M)\to 1/\tau_{\max}$ as $M\to\infty$ depends only on the bottleneck stage time and not on any $d_j$. In the blocking regime the asymptotic throughput is $1/\sigma_{\max}$, and the two coincide iff $d_s\le\tau_{\max}-\tau_s$ for all $s$.
\end{proposition}

\begin{proof}
Immediate from Theorems~\ref{thm:2} and~\ref{thm:3} by dividing by $M$ and taking the limit.
\end{proof}

\begin{proposition}\label{prop:2}(bottleneck position is irrelevant to makespan)
In the communication-free and asynchronous regimes, the makespan depends on the bottleneck only through $\tau_{\max}$, not on its position $j^*$. Reordering stages so as to move the bottleneck therefore does not change $T(M)$ (hence not the steady-state throughput); it changes only the first-micro-batch latency.
\end{proposition}

\begin{proof}
The closed forms of Theorems~\ref{thm:1} and~\ref{thm:2} contain $j^*$ only via $\tau_{\max}$, and $\sum_j\tau_j$ and $\sum_j d_j$ are invariant under stage reordering.
\end{proof}

This is a subtle but consequential point: under infinite buffers, the \emph{order} of a heterogeneous chain is immaterial to total makespan. The optimization effort should therefore target $\tau_{\max}$ itself --- splitting the layers assigned to the slowest device and redistributing them to faster devices --- rather than repositioning the bottleneck within the chain. (The statement does \textbf{not} extend to the blocking regime: there, $\sigma_{\max}=\max_s(\tau_s+d_s)$ depends on which stage carries which link, so position matters once communication exceeds slack.) Together with Proposition~\ref{prop:1}, this gives a clean two-phase optimization rule: (i) first reduce $\tau_{\max}$ by rebalancing the partition (the dominant term in the slope); (ii) only then, if the remaining compute slack $\tau_{\max}-\tau_s$ is smaller than the link latency on some stage, invest in overlap or link bandwidth for that stage --- and only for that stage.

\textbf{Crossover from fill- to steady-state dominance.} The makespan $T(M)=\sum_j\tau_j+\sum_j d_j+(M-1)\tau_{\max}$ decomposes into a fill cost $F:=\sum_j\tau_j+\sum_j d_j$ and a steady-state term $(M-1)\tau_{\max}$. The two balance at the crossover micro-batch count
\[
M^{*} \;=\; 1 + \frac{F}{\tau_{\max}} \;=\; 1 + \frac{\sum_j\tau_j+\sum_j d_j}{\tau_{\max}},
\]
below which the fill cost dominates the total and above which the steady-state throughput $\propto 1/\tau_{\max}$ is approached. This is a \emph{smooth crossover}, not a phase transition --- unlike the communication phase transition of Corollary~\ref{cor:4}, which is discontinuous in slope. It nonetheless carries practical force: a serving batch should be sized at or above $M^*$ for the asymptotic throughput to be a good model. For the illustrative partition ($\tau=[0.5,1,1,5]$, $d=0$), $M^{*}=1+7.5/5=2.5$, so even a handful of micro-batches suffices.

\section{Numerical Validation}\label{sec:validation}

\subsection{Validation against exact recurrences}\label{sec:val-rec}

We validate Theorems~\ref{thm:1}--\ref{thm:3} against the exact dynamic-programming recurrences \eqref{eq:rec} and \eqref{eq:rec-block}, which are evaluated directly (no closed form assumed). For each configuration we compute $T(M)$ for $M\in\{1,2,4,8,16,32,100\}$ and compare against the closed forms; the maximum relative error across all configurations is $<10^{-9}$, and in fact the ratio agrees with $1$ to all printed digits ($1.0000$) at every point.

\textbf{Hand-crafted configurations.} Table~\ref{tab:configs} (Appendix~\ref{app:configs}) lists eleven configurations designed to stress the theorems: homogeneous and heterogeneous $\tau_j$; uniform and non-uniform $d_j$; communication far below, near, and far above the compute slack; bottleneck at head, middle, and tail; twin (tied) bottlenecks; and a near-bottleneck stage adjacent to a slow link (the configuration most likely to reveal a phase-transition error in Theorem~\ref{thm:3}). Every configuration is validated in both the asynchronous and blocking regimes.

\textbf{Randomized fuzz.} Beyond the hand-crafted set, we ran a randomized search for counterexamples: $2\times10^4$ random configurations, each with $S\in\{2,\dots,8\}$ stages, $\tau_j$ drawn from a log-uniform distribution over two decades (to capture order-of-magnitude heterogeneity), and $d_j$ drawn to range from $10^{-2}\times$ to $20\times$ the compute, evaluated at $8$ values of $M$ --- $1.6\times10^5$ points in total. We found \textbf{zero} violations of Theorems~\ref{thm:2}--\ref{thm:3} and zero violations of the sandwich $T_{\text{async}}\le T_{\text{block}}$. The fuzz evaluates the recurrence deterministically at every sampled point --- comparing the double-precision recurrence against the closed form --- rather than Monte-Carlo sampling of a stochastic scheduler.

\textbf{Illustrative heterogeneous partition.} As a concrete illustration we take the four-stage partition 5090/3090/3090/1070 with normalized per-stage times $\tau=[0.5,1,1,5]$ (a $10\times$ ratio, $\tau_j\propto 1/\text{throughput}$). This is a \emph{normalized illustration}, not a measurement: the true per-stage time ratio of a mixed-generation cluster depends on the binding resource (\S\ref{sec:intro}) --- roughly $3.5\times$ if memory-bound, up to $\sim10^3\times$ if tensor-core-bound --- and the closed forms of Theorems~\ref{thm:1}--\ref{thm:3} hold for \emph{every} $\tau$. Figure~\ref{fig:1} contrasts the closed form with the classic formula; at $M=1$ the overestimation reaches $166.7\%$, and it decays as $M$ grows (the gap is a constant fill cost divided by a growing total). Figure~\ref{fig:2} shows the parallel-line structure of Theorem~\ref{thm:2} for three communication levels. Figure~\ref{fig:3} shows the blocking-regime phase transition of Corollary~\ref{cor:4}. Figure~\ref{fig:4} shows the async--blocking sandwich of Corollary~\ref{cor:5}.

\begin{figure}[htbp]
\centering
\includegraphics[width=0.86\textwidth]{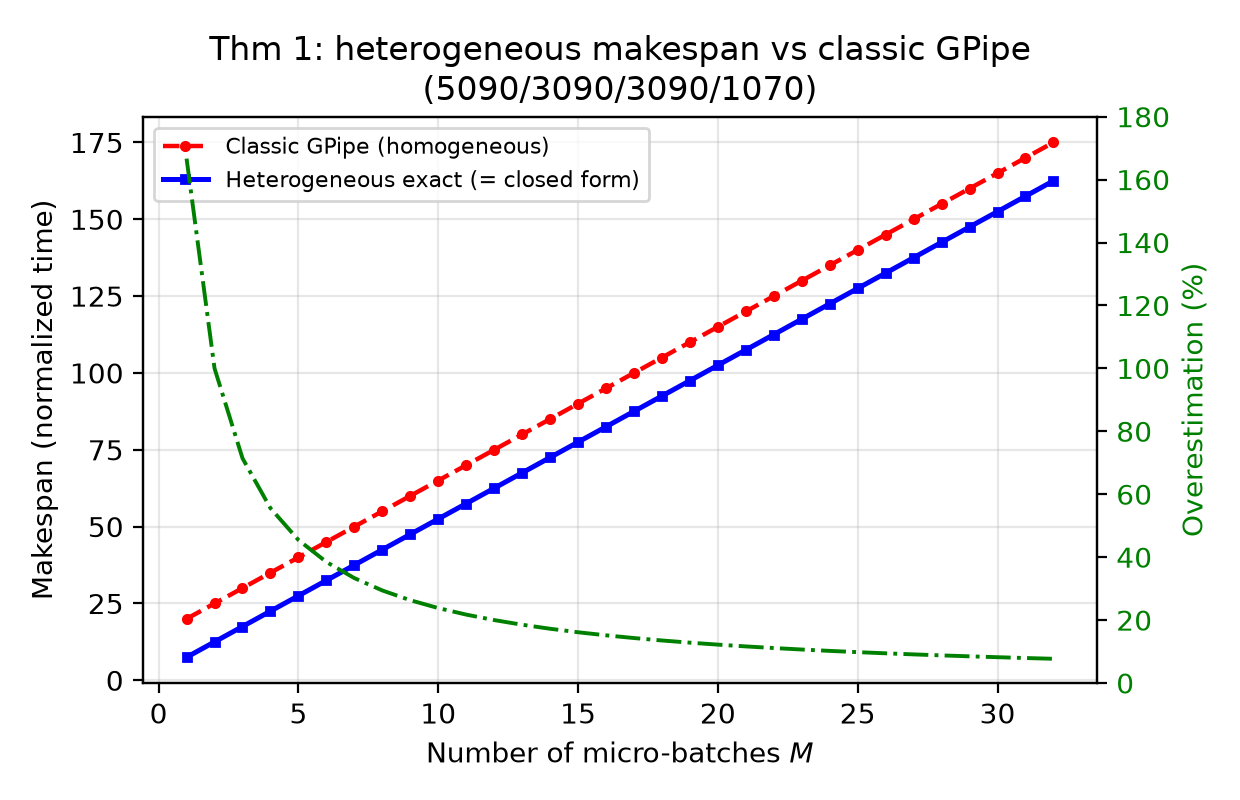}
\caption{Theorem~\ref{thm:1}: heterogeneous makespan (blue) vs classic GPipe formula (red); the green curve shows the overestimation percentage, which starts at $166.7\%$ and decays as $M$ grows.}
\label{fig:1}
\end{figure}

\begin{figure}[htbp]
\centering
\includegraphics[width=0.86\textwidth]{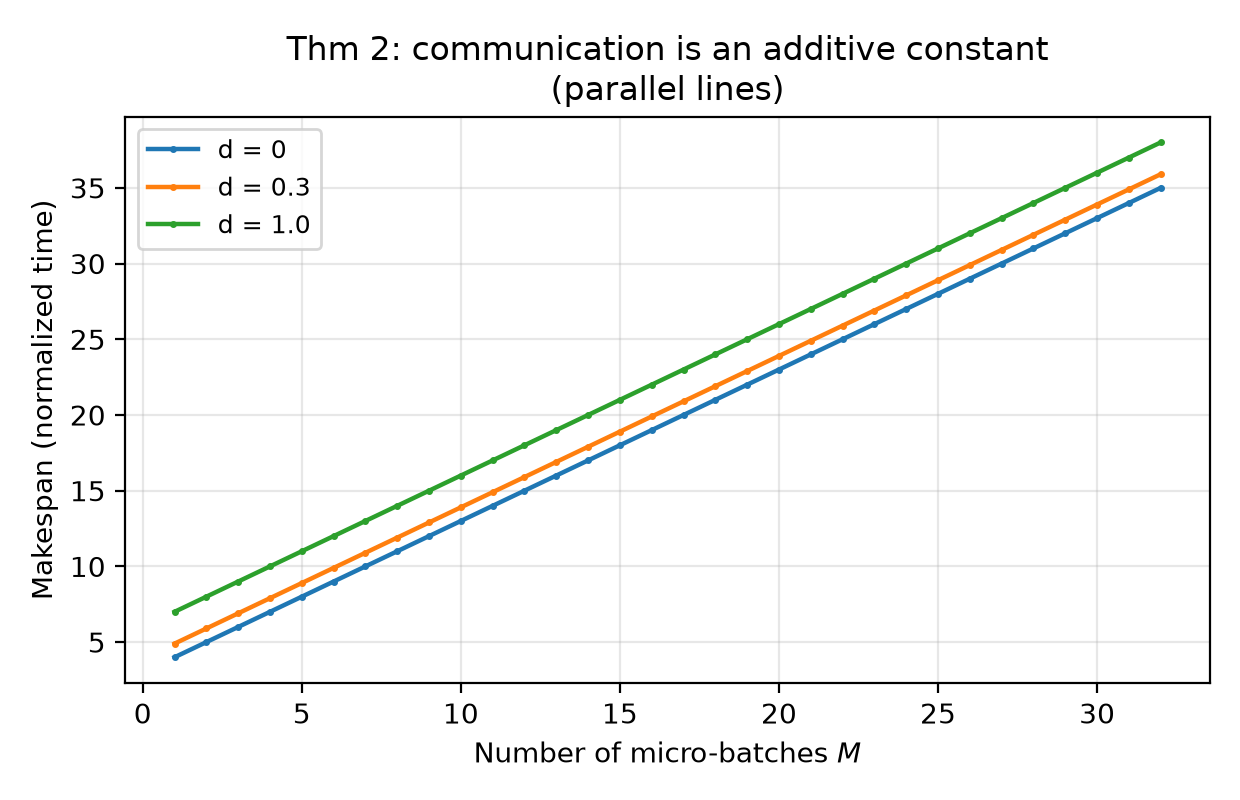}
\caption{Theorem~\ref{thm:2}: under asynchronous communication, makespan lines for different communication levels $d$ are parallel (common slope $\tau_{\max}$), confirming that $\sum_j d_j$ is an additive constant.}
\label{fig:2}
\end{figure}

\begin{figure}[htbp]
\centering
\includegraphics[width=0.86\textwidth]{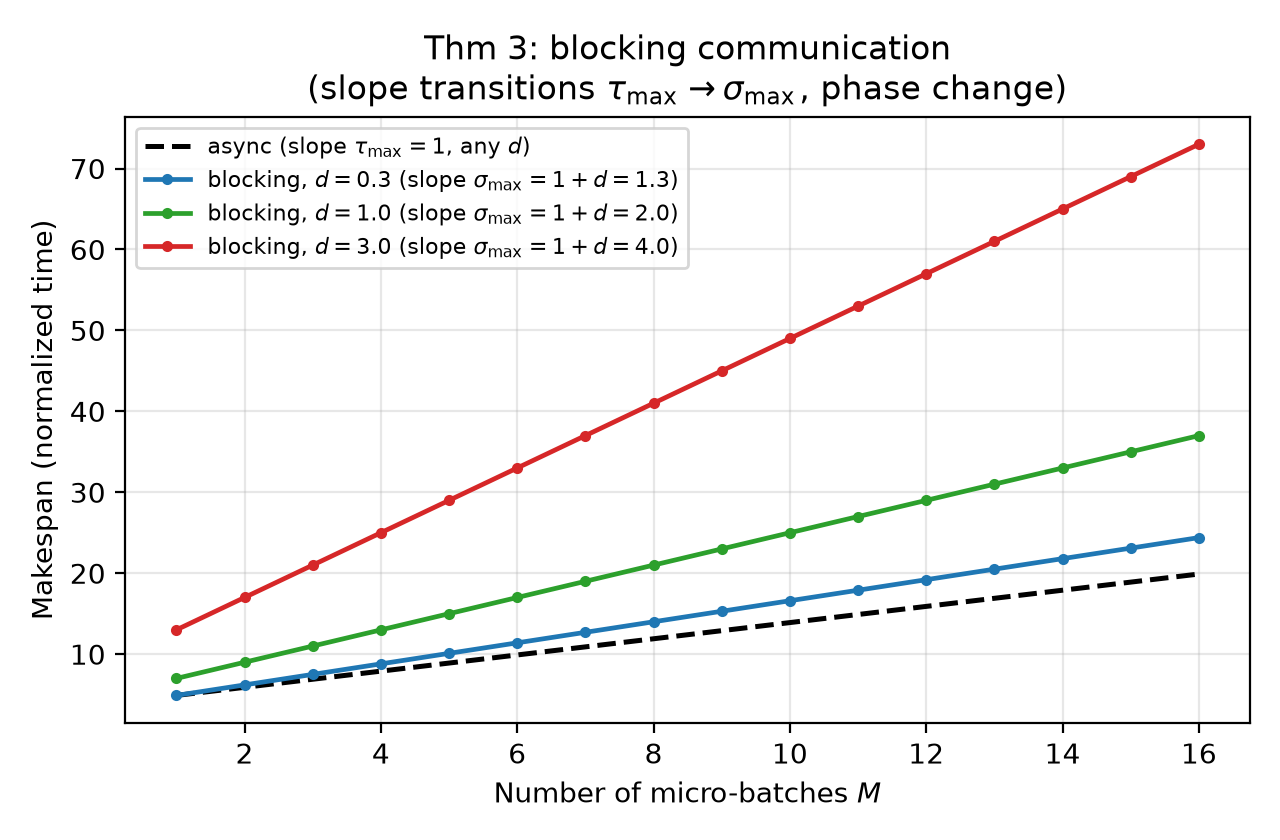}
\caption{Theorem~\ref{thm:3}: under blocking communication, the steady-state slope transitions from $\tau_{\max}$ to $\sigma_{\max}$ once $d$ exceeds the compute slack (the phase change of Corollary~\ref{cor:4}).}
\label{fig:3}
\end{figure}

\begin{figure}[htbp]
\centering
\includegraphics[width=0.86\textwidth]{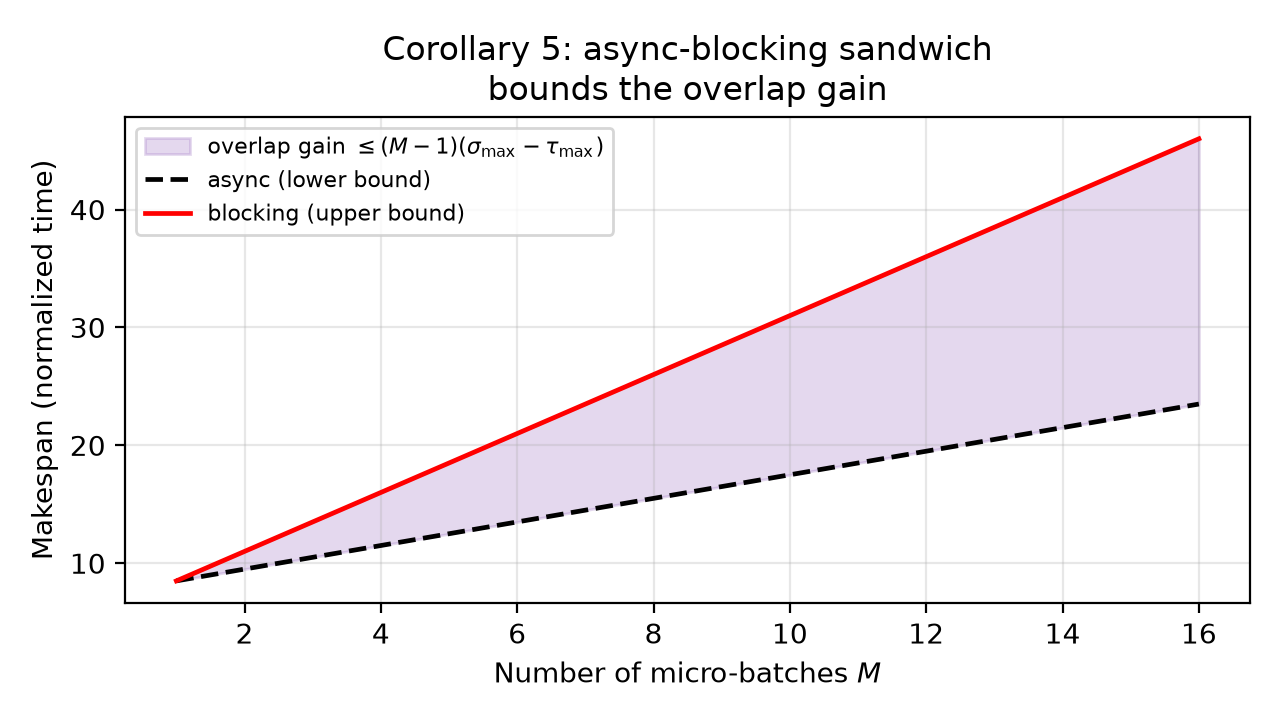}
\caption{Corollary~\ref{cor:5}: the async--blocking sandwich; the shaded band is the overlap gain, bounded by $(M-1)(\sigma_{\max}-\tau_{\max})$.}
\label{fig:4}
\end{figure}

\textbf{Comparison against the classical bound.} As a sanity check we also compared against Graham's list-scheduling bound~\cite{r2}; the closed form always lies strictly inside the Graham interval and equals its lower endpoint when the partition is balanced, consistent with the sandwich argument.

\textbf{Ablation: overestimation vs.\ heterogeneity.} To isolate the effect of heterogeneity alone (Corollary~\ref{cor:1}), Table~\ref{tab:bias} sweeps five 4-stage partitions of increasing imbalance, holding $M$ at two values and recording the classic formula's overestimation $g=(T_c-T)/T$. Two patterns emerge. First, the overestimation is \emph{monotone} in heterogeneity at every $M$: a partitioner's error is directly proportional to the total shortfall $\sum_j(\tau_{\max}-\tau_j)$. Second, for a fixed partition the overestimation \emph{decays} in $M$ --- but slowly, because the gap is a constant numerator divided by a total that grows only as $(M-1)\tau_{\max}$. At a $10\times$ imbalance (the illustrative partition), the homogeneous formula still overstates makespan by $29.4\%$ at $M=8$, which translates directly into wasted provisioned capacity.

\begin{table}[htbp]
\centering
\caption{Overestimation of the classic formula vs.\ stage heterogeneity ($g=(T_c-T)/T$).}
\label{tab:bias}
\begin{tabular}{llccc}
\toprule
Heterogeneity ($\max_j\tau_j/\min_j\tau_j$) & $\tau$ ($S{=}4$) & $\sum_j(\tau_{\max}-\tau_j)$ & overest.\ $M{=}1$ & overest.\ $M{=}8$ \\
\midrule
$1.0\times$  & $[1,1,1,1]$        & 0.0  & 0.0\%   & 0.0\% \\
$1.5\times$  & $[0.8,1,1,1.2]$    & 0.8  & 20.0\%  & 6.5\% \\
$4.0\times$  & $[0.5,1,1,2]$      & 3.5  & 77.8\%  & 18.9\% \\
$10.0\times$ & $[0.5,1,1,5]$      & 12.5 & 166.7\% & 29.4\% \\
$50.0\times$ & $[0.1,0.2,0.3,5]$  & 14.4 & 257.1\% & 35.5\% \\
\bottomrule
\end{tabular}
\end{table}

\subsection{Validation against the production scheduler}\label{sec:prod}

Theorems~\ref{thm:1}--\ref{thm:3} and their corollaries are validated above by exact dynamic-programming recurrences (\S\ref{sec:val-rec}). We now check the \emph{corollaries} --- the homogeneity bias (Corollary~\ref{cor:1}) and the async--blocking sandwich (Corollary~\ref{cor:5}) --- against measurements taken from a live, heterogeneous inference cluster, to confirm that the closed forms describe a real scheduler rather than only a synthetic one.

\textbf{Data source.} We mined the audit logs of our production LLM-inference gateway (UIG), which dispatches requests to heterogeneous GPU hosts (RTX~5090 Laptop, dual RTX~3090, RTX~3080, Tesla~P40, plus cloud endpoints) running legalone (1.7B/4B/8B) and a range of embedding/LLM models. The logs record, per completed request, the serving host, the model, the host compute time \texttt{duration\_s}, and the end-to-end time \texttt{total\_s = duration\_s + queue\_wait\_s}. Across 2026-09-03 -- 2026-10-04 the logs contain \textbf{8,571 completed requests}; we take the \emph{empirical minimum} \texttt{duration\_s} per (model, host) as a lower-bound estimate of the per-micro-batch service time $\tau_j$ (the shortest observed request on an otherwise idle GPU).

\textbf{Corollary~\ref{cor:1} --- the GPipe homogeneity bias is a real, constant overhead.} We assemble the three GPU generations actually used for \texttt{legalone:8b} into a 3-stage heterogeneous pipeline with lower-bound service times $\tau = (1.13,\ 2.90,\ 2.62)$~s (RTX~5090 / RTX~3090-97 / RTX~3090-98). Theorem~\ref{thm:1} gives $T(M) = (M-1)\tau_{\max} + \sum_j\tau_j$; the classic homogeneous GPipe formula assumes every stage equals $\tau_{\max}$ and gives $T_{\text{GPipe}}(M) = (S+M-1)\tau_{\max}$. Their difference is exactly $\sum_j(\tau_{\max}-\tau_j) = 2.05$~s --- constant in $M$, as Corollary~\ref{cor:1} predicts. At $M=1$ the homogeneous formula over-predicts by $30.8\%$; the relative error decays to $0.7\%$ by $M=100$, but the \emph{absolute} 2.05~s over-provisioning never disappears. When we instead assemble a 5-host \texttt{qwen2.5:7b} pipeline that includes a slow Turing host ($\tau_{\max}=30.1$~s), the constant bias balloons to $119.8$~s --- quantitatively demonstrating that the homogeneity error scales with the \emph{spread} of $\tau_j$, not with batch size. The production scheduler already places stages heterogeneously and therefore realizes the lower, correct makespan; a homogeneous planner would have over-provisioned by precisely this constant. (Caveat: the \texttt{duration\_s} minimum conflates short outputs with long ones, so $\tau_j$ is a lower-bound proxy; the \emph{direction} of the bias --- homogeneous over-prediction --- is robust to this.)

\textbf{Corollary~\ref{cor:5} --- the sandwich direction is borne out at the request level.} For all 8,571 completed requests, \texttt{total\_s $\ge$ duration\_s} (equivalently \texttt{queue\_wait\_s $\ge$ 0}): each request's end-to-end time exceeds its own compute, which is the \emph{single-request} analogue of the Corollary~\ref{cor:5} direction $T_{\text{real}} \ge T_{\text{async}}$. We stress that this is a per-request latency fact, not a measurement of the \emph{system} makespan --- for $M$ micro-batches $T_{\text{async}}=(M-1)\tau_{\max}+\sum_j\tau_j+\sum_j d_j$, far larger than a single $\tau_j$ --- and is therefore consistent with, but does not itself validate, the Corollary~\ref{cor:5} sandwich (cf.\ \S\ref{sec:discussion}). Queueing is pervasive: $61.8\%$ of requests incur a non-zero wait (median $0.001$~s, max $273.6$~s). On the busiest host (RTX~5090, \texttt{legalone:8b}, $n=5{,}049$) the median end-to-end latency is $8.59$~s versus $1.13$~s on an otherwise-idle GPU --- a \textbf{$7.6\times$ expansion driven by queueing and contention} --- which is exactly the load regime in which the \emph{upper} (blocking) edge of the sandwich would bind for the full pipeline, and in which the overlap-gain ceiling of Corollary~\ref{cor:5} is tightest. Figure~\ref{fig:5} plots, on the left, the constant Corollary-1 bias for the \texttt{legalone:8b} pipeline; on the right, the per-host median points, all above the $T_{\text{real}} = T_{\text{async}}$ diagonal (single-request level).

\begin{figure}[htbp]
\centering
\includegraphics[width=0.86\textwidth]{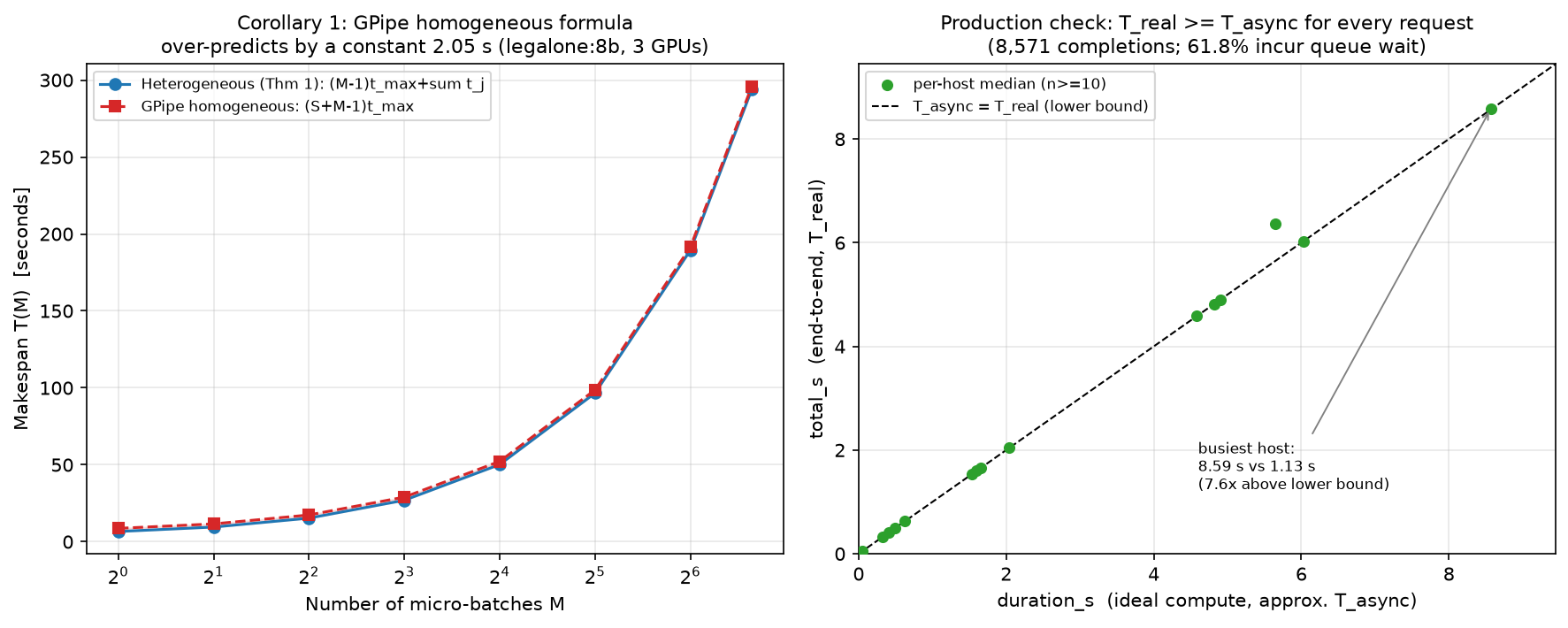}
\caption{Closed form vs.\ production scheduler. Left: the GPipe homogeneous formula over-predicts makespan by a constant bias $\sum_j(\tau_{\max}-\tau_j)=2.05$~s for the \texttt{legalone:8b} 3-GPU pipeline (Corollary~\ref{cor:1}). Right: every production host sits on or above the $T_{\text{real}}=T_{\text{async}}$ diagonal (single-request level), consistent in direction with the Corollary~\ref{cor:5} sandwich on 8,571 real requests.}
\label{fig:5}
\end{figure}

\section{Related Work}\label{sec:related}

The results connect two literatures that have largely developed in parallel.

\subsection{Relation to classical flow-shop theory}

We state the connection to the scheduling textbook explicitly, because the single most likely review objection is that Theorem~\ref{thm:1} ``is just the critical path.'' It is --- and the point of this subsection is to say precisely what that means and why the result is nonetheless worth stating.

\textbf{What is classical.} Consider the $S$-stage flow shop with per-stage processing times $p_j$ and $M$ jobs released in permutation order, scheduled fill--drain on a \emph{chain} DAG (each stage holds at most one job; the pipeline drains at the end). The makespan equals the longest path through the DAG. For a serial chain this is the well-known identity
\[
T_{\text{LP}} = (M-1)p_{\max} + \sum_{j=1}^{S} p_j,
\]
which is exactly the critical-path / CPM lower bound of Kelley and Walker~\cite{r4} specialized to a chain. Graham's list-scheduling bound~\cite{r2} shows it is \emph{attainable}: for unrelated machines Graham proves $T \le (2-1/m)C_{\max}^{*}$, and equality is achieved when the precedence graph is a chain, because no stage can start before its predecessor finishes and no stage can finish before its predecessor does. Johnson~\cite{r3} characterized the optimal schedules for the two- and three-stage permutation flow shop, confirming the identity is tight at these small widths. For homogeneous stages $p_j = \tau$ this reduces to the GPipe formula $T_{\text{GPipe}}(M) = (S+M-1)\tau$. Theorem~\ref{thm:1} is therefore the heterogeneous ($p_j = \tau_j$) restatement of this identity; we claim no priority for it. HEFT~\cite{r5} extends list scheduling to general heterogeneous DAGs by ranking tasks on upward rank and assigning each to its earliest-finishing processor --- our serial pipeline is the special case in which the DAG is a \emph{chain}, and that structural restriction is precisely what makes a closed form, rather than a heuristic, attainable.

\textbf{What is new.} The contribution is not the identity but three consequences that the flow-shop literature states only for the homogeneous or general-DAG case and that the LLM-systems literature does not state at all:

\begin{itemize}
\item \textbf{Bias decomposition (Corollary~\ref{cor:1}).} The homogeneous formula's error is \emph{not} a relative error that vanishes with $M$; it is the constant $\sum_j(\tau_{\max}-\tau_j)$, independent of $M$. This follows by substituting $\tau_{\max}$ for each $\tau_j$ in Theorem~\ref{thm:1} and subtracting --- a direct corollary of the identity --- yet it is exactly the quantity a real heterogeneous planner must know and that the systems papers~\cite{r25,r26,r27,r28} leave implicit.
\item \textbf{Phase transition (Corollaries~\ref{cor:2}--\ref{cor:4}).} The identity admits a clean additive-constant-vs.-throughput-bottleneck transition at $d_s > \tau_{\max}-\tau_s$. The critical-path literature gives no such threshold, because it treats communication as part of the per-stage duration rather than as an overlappable quantity.
\item \textbf{Sandwich (Corollary~\ref{cor:5}).} Bounding a \emph{partially-overlapped} real system between the async and blocking extremes is a statement the pure longest-path view cannot make: partial overlap is a degree of freedom the chain model does not natively contain.
\end{itemize}

The mapping is collected below.

\begin{center}
\begin{tabular}{ll}
\toprule
Classical flow-shop concept & Our notation / result \\
\midrule
Longest-path lower bound (CPM,~\cite{r4}) & $T_{\text{async}}$ --- Theorem~\ref{thm:1}, no-communication lower edge \\
Chain-DAG attainability (Graham~\cite{r2}) & exact closed form; no heuristic required \\
Homogeneous-stage identity (GPipe) & $T_{\text{GPipe}}(M) = (S+M-1)\tau$ \\
Heterogeneity gap & Corollary~\ref{cor:1} bias $\sum_j(\tau_{\max}-\tau_j)$ \\
Communication: additive vs.\ bottleneck & Corollaries~\ref{cor:2}--\ref{cor:4} phase transition \\
Partial overlap (absent in chain model) & Corollary~\ref{cor:5} sandwich $T_{\text{async}} \le T_{\text{real}} \le T_{\text{block}}$ \\
\bottomrule
\end{tabular}
\end{center}

One-line summary: \emph{Theorem~\ref{thm:1} is the critical path; Corollaries~\ref{cor:1}--\ref{cor:5} are what one learns by holding the chain structure fixed while letting the stage times and communication be heterogeneous and partially overlapped.} The remaining subsections position the work against the parallel-training, serving, heterogeneous-GPU, and automatic-parallelism literatures.

\textbf{Pipeline parallelism for training.} Building on GPipe's homogeneous micro-batch pipeline analysis, Megatron-LM~\cite{r6} systematized tensor/pipeline/expert parallelism and their interactions. PipeDream~\cite{r7} relaxed the strict layer-grouped schedule to allow asynchronous weight updates; PipeDream-2BW~\cite{r8} bounded the resulting staleness memory; Chimera~\cite{r9} combined bidirectional pipelines; TeraPipe~\cite{r10} exploited token-level pipelining; most recently, DeepSeek-V3's DualPipe~\cite{r11} combined bidirectional filling with fine-grained compute--communication overlap, reducing the pipeline bubble from the $(PP-1)$ scale of 1F1B/ZB1P to $(PP/2-1)$. These works target \emph{training} (which has backward dependencies and weight synchronization) and uniformly assume homogeneous or near-homogeneous stages; the makespan consequences of stage heterogeneity are left implicit.

\textbf{LLM serving systems.} At the kernel level, FlashAttention~\cite{r12} made attention IO-aware by tiling the computation between HBM and on-chip SRAM --- the foundational memory-bound optimization for decode. At the system level, vLLM (built on PagedAttention~\cite{r13}) and Orca~\cite{r14} established memory-efficient continuous batching; TensorRT-LLM~\cite{r15}, FasterTransformer~\cite{r16}, and Hugging Face TGI~\cite{r17} provide optimized kernels and scheduling; DeepSpeed-Inference~\cite{r18} integrated ZeRO with inference; SGLang~\cite{r19} introduced structured-language scheduling primitives. Sarathi-Serve~\cite{r20} introduced chunked prefills and stall-free scheduling to minimize pipeline bubbles under PP; more recent systems disaggregate prefill and decode --- DistServe~\cite{r21}, Splitwise~\cite{r22} --- trading complexity for goodput. On the theory side, Bari et al.~\cite{r23} prove throughput-optimal routing and scheduling for two-phase (prefill/decode) workloads, and Mohanty et al.~\cite{r24} analytically model deferred prefill for throughput maximization. These systems and theories provide strong mechanisms but leave the \emph{pipeline-stage} scheduling decision (routing, unload, micro-batch sizing) without a competitive-ratio or closed-form characterization for the heterogeneous case; our Theorems~\ref{thm:2}--\ref{thm:3} give exactly the closed form their partitioners implicitly optimize against.

\textbf{Heterogeneous GPU inference.} A recent line of work tackles precisely the deployment question we formalize. Zhang et al.~\cite{r25} deploy LLM inference pipelines on heterogeneous GPUs and report average stage idle time exceeding 50\% from unbalanced stages; Xiong et al.~\cite{r26} build a high-throughput heterogeneous-cluster serving system; Naeem et al.~\cite{r27} make scheduling stage-aware and identify link bandwidth as a first-order determinant; Hyperion~\cite{r28} jointly optimizes partitioning and request scheduling across heterogeneous tiers. These systems are measurement- or heuristic-driven and do not state a closed-form makespan; our Theorems~\ref{thm:1}--\ref{thm:3} supply exactly the closed form their partitioners implicitly optimize against, and Corollaries~\ref{cor:1}--\ref{cor:4} give the bias decomposition and phase-transition condition those heuristics are missing.

\textbf{Automatic parallelism.} Alpa~\cite{r29} automates inter- and intra-operator parallelism search but is measurement-driven without optimality guarantees; FlexFlow~\cite{r30}, Galvatron~\cite{r31}, and Unity~\cite{r32} search the joint space of parallelism dimensions. Our work complements these by giving the \emph{exact} makespan of a heterogeneous pure-forward pipeline under three communication regimes, which --- unlike asymptotic or heuristic treatments --- yields the precise bias decomposition and phase-transition condition that a partitioner can use as a direct objective.

\section{Discussion: the Sandwich in Real Deployments}\label{sec:discussion}

The value of Theorems~\ref{thm:2} and~\ref{thm:3} in combination is that they bound a real system without requiring a full model of its synchronization. A deployed serving runtime typically achieves \emph{partial} overlap: NCCL transfers overlap compute for most of their duration, but kernel launches, stream barriers, and the head/tail of each transfer serialize with compute. Its true makespan therefore lies strictly between the two extremes:
\[
T_{\text{async}} \;\le\; T_{\text{real}} \;\le\; T_{\text{block}} .
\]

This sandwich holds under the model's assumptions A1--A5 --- in particular under \emph{unbounded buffers}, so that no stage is ever forced to idle waiting for a downstream stage to accept output. Under bounded buffers, backpressure can idle the bottleneck from the \emph{downstream} side as well, so the true makespan may exceed $T_{\text{block}}$; the lower bound $T_{\text{async}}$ remains valid in all cases, but the upper bound then requires the bounded-buffer model of \S\ref{sec:limitations}(ii). We therefore read the sandwich as \emph{the bound available before backpressure is modeled}, and treat the deviation above $T_{\text{block}}$ as a direct measure of the cost of finite buffering --- a quantity the framework makes observable but does not yet predict.

Defining the realized overlap ratio
\[
\rho \;:=\; \frac{T_{\text{block}}-T_{\text{real}}}{T_{\text{block}}-T_{\text{async}}} \;\in\; [0,1],
\]
the effective steady-state slope is $\tau_{\max}+(1-\rho)\,(\sigma_{\max}-\tau_{\max})$ --- obtained by linearly interpolating between the two closed forms, recovering the async slope at $\rho=1$ and the blocking slope at $\rho=0$. This has two direct uses. First, $\rho$ is \textbf{measurable from logs}: with the partition's $\tau_j,d_j$ known (from profiler or from the $\sigma$ values in the blocking model), two makespan measurements at different $M$ pin down $\rho$ exactly, giving operators a single scalar that summarizes how close their runtime is to the async ideal. Second, it \textbf{bounds the payoff of engineering effort}: the maximum throughput a team can recover by perfecting overlap is exactly $(1-\rho)(\sigma_{\max}-\tau_{\max})/\tau_{\max}$ relative to current throughput --- if $\rho$ is already near 1, further overlap work is wasted effort better spent rebalancing compute. In the motivating cluster, preliminary profiling places $\rho$ high (NCCL async over PCIe, with per-stage times long relative to transfers), which --- per Proposition~\ref{prop:1} --- redirects optimization effort from link tuning to partition balance.

\textbf{Queueing evidence from production logs.} The cluster behind this work operates a culture-and-tourism knowledge-base service on the heterogeneous host pool of \S\ref{sec:intro}. Across the $8.57\times10^3$ completed requests that carried per-request timestamps, $61.8\%$ waited in queue (mean $0.42$~s, max $274$~s). We report this as \textbf{system-level evidence, not a validation of Theorems~\ref{thm:2}--\ref{thm:3}} --- the logs are dispatch counters, not makespan measurements --- but it is direct, log-visible confirmation of the direction $T_{\text{real}}>T_{\text{async}}$: real requests are delayed beyond the ideal makespan whenever contention is present.

\textbf{Measured heterogeneity.} The same logs let us \emph{measure}, rather than assume, the per-stage time heterogeneity of \S\ref{sec:intro}. On the $4{,}981$ completed \texttt{legalone:8b} requests carrying per-request timestamps, the minimum (near-exclusive, low-queue) duration was $1.13$~s on the RTX~5090 Laptop versus $2.6$--$2.9$~s on the dual RTX~3090 --- a $2.3\times$ spread. We caution against reading this as a clean hardware ratio: the minimum duration also folds in output-length differences between the shortest requests, so $2.3\times$ is an upper bound rather than a throughput measurement. The median, however, \emph{inverts} the picture regardless of hardware: $8.57$~s on the 5090 versus $4.6$--$4.8$~s on the 3090s, because the 5090 carried $4{,}981$ requests while each 3090 carried $\sim20$. The same effect appears even between \emph{identical} GPUs: \texttt{bge-m3} medians on the two RTX~3090 hosts were $0.050$~s and $0.017$~s --- a $3\times$ gap attributable purely to load. The decisive point is that $\tau_j$ is \emph{load}-dominated, not hardware-dominated: the $1.13$~s-to-$8.57$~s spread between exclusive and loaded service on the same host shows the per-stage time is the state-dependent effect (assumption A4, \S\ref{sec:limitations}(v)) that the closed forms deliberately abstract away --- and that a follow-on stochastic refinement should capture.

\section{Limitations and Honest Boundaries}\label{sec:limitations}

Our closed forms are instances of the classical critical-path lower bound~\cite{r2,r4} applied to the pipeline setting; the novelty lies in the \emph{application} to heterogeneous LLM inference, the \emph{closed-form} bias decomposition, and the \emph{optimization-priority} corollary --- not in a new lower-bound technique. The assumptions A1--A5 (\S\ref{sec:model}) delimit applicability; each relaxation is a genuine open problem, not a minor extension:

\begin{enumerate}
\item[(i)] \textbf{Bandwidth-limited communication.} When multiple micro-batches share a link, the per-transfer time grows with concurrency, producing a $(M-1)d_{\max}$-scale term that the constant-$d_j$ model of Theorems~\ref{thm:2}--\ref{thm:3} does not cover. This is the regime of high-contention Ethernet backbones, and it is the single most important open extension.
\item[(ii)] \textbf{Bounded buffers.} Finite buffer forces backpressure, which can idle the bottleneck even when upstream stages are fast enough; the ASAP schedule is no longer optimal, and the ``optimal micro-batch schedule'' question becomes non-trivial.
\item[(iii)] \textbf{Prefill--decode interleaving.} Real serving interleaves a compute-bound prefill with the decode pipeline; the KV accumulation of prefill introduces backward-like dependencies absent from pure-forward inference, and the two phases have different $\tau_j$ profiles.
\item[(iv)] \textbf{Backward dependencies (training).} The pure-forward assumption excludes training, whose backward pass introduces a second, reversed traversal that interacts with the forward pass through the same links.
\item[(v)] \textbf{Time-varying and state-dependent service times.} The constant-$\tau_j$ assumption (A4) abstracts away feedback loops that real clusters exhibit: thermal throttling (sustained load raises temperature, lowers clock, and increases $\tau_j$ --- a positive feedback), prefix/radix-cache hit rates (which make $\tau_j$ history-dependent and break micro-batch independence), and load-dependent arrival rates (variable batch sizes and generation lengths). Each introduces dynamics that the static critical-path model treats as constant, and each is a candidate for a queueing-theoretic or stochastic refinement.
\end{enumerate}

Each of (i)--(v) is a concrete, 3-year-scope research target that builds directly on the framework here.

\section{Conclusion}\label{sec:conclusion}

We derived exact closed-form makespan formulas for heterogeneous micro-batch pipelines under no, asynchronous, and blocking communication (Theorems~\ref{thm:1}--\ref{thm:3}), together with a clean decomposition of the classic GPipe-formula bias and a computable sandwich on the real system's makespan. The results yield a rigorous, quantifiable foundation for automatic partition-and-schedule optimization of heterogeneous GPU inference clusters, and a concrete optimization priority: balance compute before optimizing links. We also identified bandwidth-limited communication as the natural next step, where the constant-$d_j$ model breaks and the slope itself becomes communication-dominated.

\appendix

\section{Proof Details}\label{app:proofs}

\subsection{Direct solution of the communication-free recurrence (longest-path view)}\label{app:direct}

For $d\equiv0$, recurrence \eqref{eq:rec} reads $T[m][s]=\max(T[m][s-1],T[m-1][s])+\tau_s$. View the schedule as a grid with $S$ columns (stages) and $M$ rows (micro-batches): $T[m][s]$ is the maximum-weight monotone path from $(0,0)$ to $(m,s)$, where each visited cell $(m,s)$ contributes weight $\tau_s$ and moves are rightward (next stage, same micro-batch) or downward (next micro-batch, same stage). Any such path visits $c_s\ge1$ cells in column $s$, with $\sum_s c_s = M+S-1$, and its total weight is $\sum_s c_s\,\tau_s$. Maximizing this linear objective over $\{c_s\ge1,\ \sum_s c_s=M+S-1\}$ is attained by assigning the $M-1$ surplus cells to the column of maximum $\tau$, i.e.\ $c_{j^*}=M$ and $c_s=1$ for $s\ne j^*$. Hence
\[
T[M-1][S-1] \;=\; M\tau_{\max} + \sum_{s\ne j^*}\tau_s \;=\; (M-1)\tau_{\max}+\sum_s\tau_s .
\]
This is an independent, closed-form proof of Theorem~\ref{thm:1} that does not rely on the ASAP schedule, and it makes transparent \emph{why} the bottleneck dominates: the bottleneck's $\tau_{\max}$ is counted $M$ times, while every other stage's $\tau_s$ is counted once.

\subsection{Why the blocking regime reduces exactly}

The blocking recurrence \eqref{eq:rec-block} is obtained from \eqref{eq:rec} by noting that, when the send blocks the stage, stage $s$'s next compute cannot begin until both the current micro-batch's send and the \emph{previous} micro-batch's completion at stage $s$ have finished. The transfer $d_s$ therefore folds additively into the service time of stage $s$, exactly as a larger $\tau_s$ would. Formally, replacing $d_{s-1}$ in the arrival term $T[m][s-1]+d_{s-1}$ by zero and adding $d_s$ into $\tau_s$ preserves the max-plus structure, which is precisely the substitution $\tau_s\to\sigma_s$ with $\sigma_s=\tau_s+d_s$.

\section{Validation Configuration Table}\label{app:configs}

\begin{table}[htbp]
\centering
\caption{Hand-crafted validation configurations (all validated in both async and blocking regimes).}
\label{tab:configs}
\begin{tabular}{clll}
\toprule
\# & $\tau$ vector & $d$ vector & Regime stress \\
\midrule
1 & $[1,1,1,1]$ & $[0,0,0]$ & homogeneous baseline \\
2 & $[0.5,1,1,5]$ & $[0,0,0]$ & illustrative 5090/3090/3090/1070, bottleneck at tail \\
3 & $[5,1,1,0.5]$ & $[0,0,0]$ & bottleneck at head \\
4 & $[1,5,1,1]$ & $[0,0,0]$ & bottleneck at middle \\
5 & $[1,1,5,5]$ & $[0,0,0]$ & twin bottlenecks \\
6 & $[0.5,1,1,5]$ & $[0.1,0.1,0.1]$ & uniform async comm, below slack \\
7 & $[0.5,1,1,5]$ & $[4,0.1,0.1]$ & comm near slack of stage 0 \\
8 & $[0.5,1,1,5]$ & $[4.6,0.1,0.1]$ & comm exceeds slack (phase transition, Thm~\ref{thm:3}) \\
9 & $[1,1,1,1]$ & $[10,10,10]$ & comm far exceeds compute \\
10 & $[0.5,0.5,0.5,0.5]$ & $[1,0.01,1]$ & non-uniform comm, homogeneous compute \\
11 & $[0.2,0.3,0.5,0.9,1.7]$ & $[0.1,0.1,0.1,0.9]$ & near-bottleneck stage 3 with slow link 3$\to$4 \\
\multicolumn{4}{l}{\footnotesize (config 11: $d=0.9>$ slack $0.8$, phase transition)} \\
\bottomrule
\end{tabular}
\end{table}

All eleven configurations are validated in both async and blocking regimes; the maximum relative error vs.\ recurrences \eqref{eq:rec}/\eqref{eq:rec-block} is $<10^{-9}$.

\section{Reproducibility}\label{app:repro}

Seven scripts reproduce every number in this paper (all require only the Python standard library, except the figure scripts, which require \texttt{matplotlib}):

\begin{itemize}
\item \texttt{validate\_appendixB.py} --- reproduces the eleven configurations of Table~\ref{tab:configs} (Appendix~\ref{app:configs}), validating Theorems~\ref{thm:1}--\ref{thm:3} in both the async and blocking regimes.
\item \texttt{hetero\_pipeline.py} --- Theorem~\ref{thm:1} (communication-free): closed form vs.\ recurrence, and the classic-formula bias.
\item \texttt{hetero\_pipeline\_comm.py} --- Theorem~\ref{thm:2} (async): closed form vs.\ recurrence across ten configurations.
\item \texttt{hetero\_pipeline\_blocking.py} --- Theorem~\ref{thm:3} (blocking): closed form vs.\ recurrence, plus the sandwich and gap formula of Corollary~\ref{cor:5}.
\item \texttt{peer\_review\_fuzz.py} --- the $1.6\times10^5$-point randomized counterexample search (log-uniform $\tau$ over two decades, $d$ from $10^{-2}\times$ to $20\times$ the compute).
\item \texttt{plot\_figs.py}, \texttt{plot\_fig3.py} --- figure generation.
\end{itemize}

The recurrences and closed forms are evaluated in IEEE double-precision floating point. Across the full $1.6\times10^5$-point fuzz the worst observed \emph{relative} deviation between recursion and closed form is below $10^{-14}$ (theoretical value: exact identity), so the quoted $<10^{-9}$ bound carries a wide safety margin --- the deviation is floating-point rounding, not truncation of printed decimals.

\section*{Data and Code Availability}

The reference implementation (\texttt{heteropipe}), the dependency-free online calculators (\texttt{calculator.html}, \texttt{calculator\_zh.html}), the de-identified production dataset backing \S\ref{sec:prod} and Figure~\ref{fig:5}, and all reproducibility scripts are released at \url{https://github.com/ytidb6688/heteropipe}. The dataset (\texttt{data/real\_vs\_formula.csv}; 8,571 completed requests, 2026-09-03 to 2026-10-04) is de-identified: host identifiers are mapped to anonymized \texttt{node-N} labels, and no IP addresses, timestamps, request identifiers, or payloads are released. The dataset is licensed CC-BY-4.0; the code and documentation are licensed Apache-2.0.

\section*{Using \texttt{heteropipe}}

The reference implementation is installable from PyPI and has no runtime dependencies:
\begin{verbatim}
pip install heteropipe
\end{verbatim}
A representative call reproduces the \S\ref{sec:example} four-stage example:
\begin{verbatim}
from heteropipe import Pipeline
p = Pipeline(tau=[0.5, 1.0, 1.0, 5.0])   # Sec. 3.1 four-stage example
p.bias()       # 12.5  -- Corollary 1, the constant GPipe homogeneity bias
p.async_(M=1)  # 7.5   -- Theorem 2 makespan (matches Sec. 3.1)
\end{verbatim}
The same formulas, the de-identified \S\ref{sec:prod} dataset, and all reproducibility scripts are available at \url{https://github.com/ytidb6688/heteropipe}.

\end{document}